\documentclass[11pt,letterpaper]{article}
\usepackage[T1]{fontenc}
\usepackage{times}
\usepackage[margin=1in]{geometry}
\usepackage{amsmath,amssymb,amsthm}
\usepackage{microtype,needspace}
\usepackage[hidelinks]{hyperref}
\hypersetup{pdftitle={List Recovery for Random Low-Rate Linear Codes},
  pdfauthor={Isaac M Hair and Amit Sahai}}

\newtheorem{theorem}{Theorem}[section]
\newtheorem{lemma}[theorem]{Lemma}
\newtheorem{proposition}[theorem]{Proposition}
\theoremstyle{definition}
\newtheorem{definition}[theorem]{Definition}
\newcommand{\Fp}{\mathbb{F}_p}
\newcommand{\MPS}{Matou\v{s}ek, P\v{r}\'{\i}v\v{e}tiv\'{y}, and \v{S}kovro\v{n}}

\title{List Recovery for Random Low-Rate Linear Codes}
\author{Isaac M Hair\thanks{\texttt{isaacmhair@gmail.com}}\\UCSB, UCLA
  \and Amit Sahai\thanks{\texttt{sahai@cs.ucla.edu}}\\UCLA}
\date{}

\begin{document}
\pagenumbering{roman}
\maketitle
\thispagestyle{empty}

\begin{abstract}
We show that random linear codes of fixed dimension $d$ over sufficiently large prime fields
have nearly optimal output list size for list recovery. Fix an error fraction $0<\alpha<1$ and
$0<\varepsilon<1$. For all sufficiently large block lengths $n$, with probability at least
$1-\varepsilon$, the following holds simultaneously for every input list size $\ell\le 2^{\delta
n/\log n}$, where $\delta>0$ depends only on $\alpha,\varepsilon$, and $d$. Given any list of
$\ell$ symbols at each coordinate, at most $(1+\varepsilon)\ell/(1-\alpha)$ codewords agree with
the lists in at least $(1-\alpha)n$ coordinates.

Our proof adapts the colored-graph and determinant arguments of \MPS{} for reconstructing planar
point sets from their projections. We also prove that the range of input list sizes is nearly
optimal. There is a constant $c=c(\alpha,\varepsilon)>0$ such that, for every linear code of
dimension at least two over a sufficiently large prime field $\Fp$ and every $2^{cn}\le\ell\le
p$, there are lists of size $\ell$ and a set of more than $(1+\varepsilon)\ell/(1-\alpha)$
codewords, each of which agrees with those lists in every coordinate.

\emph{The main results of this paper were discovered by ChatGPT 5.5 Sol using harness components from the UCLA Moonshot Harness Project, in response to a less optimal result found by the authors. See the end of the introduction for a discussion.}
\end{abstract}
\clearpage
\pagenumbering{arabic}

\section{Introduction}

List recovery asks for the codewords that agree with prescribed lists of symbols in many
coordinates. It arises in concatenated decoding when the inner decoder returns several
candidates for each block and the outer decoder must recover the original
codeword~\cite{ZyablovPinsker,HemenwayRonZewiWootters}. List recovery also underlies
interpolation-based decoding of Reed--Solomon codes~\cite{GuruswamiSudan}. We study the largest
possible output list as a function of the input list size.

\begin{definition}[List Recovery]\label{def:list-recoverable}
Let $C\subseteq\Sigma^n$ be a code over a finite alphabet $\Sigma$. For $0\le\alpha<1$, an input
list size $1\le\ell\le |\Sigma|$, and $L\ge 0$, we say that $C$ is \emph{$(\alpha,\ell,L)$-list
recoverable} if for every choice of sets $S_1,\ldots,S_n\subseteq\Sigma$ of size $\ell$, at most
$L$ codewords $x\in C$ satisfy
\[
  \bigl|\{i\in\{1,\ldots,n\}:x_i\in S_i\}\bigr|\ \ge\ (1-\alpha)n .
\]
We call $L$ the output list size.
\end{definition}

A counting argument gives a lower bound of approximately $\ell/(1-\alpha)$ on the output list
size. Set $s=\lceil(1-\alpha)n\rceil$, and suppose the code contains at least
$\lfloor n\ell/s\rfloor$ distinct codewords. Write the coordinate indices $1,\ldots,n$ in order
$\ell$ times. Assign the first $s$ indices to the first codeword, the next $s$ to the second,
and so on. Each codeword gets $s$ distinct coordinates, and each coordinate is assigned to at
most $\ell$ codewords. We can therefore choose lists of size $\ell$ so that all
$\lfloor n\ell/s\rfloor$ codewords agree with the lists in at least $s$ coordinates.

We study how closely a \emph{linear} code can approach this bound when $\ell$ grows with $n$. Over
a prime field $\Fp$, a linear code $C\subseteq\Fp^n$ of dimension $d$ has rate $d/n$. We fix $d$
and allow $p$ to grow with $n$, so that the rate tends to zero while the input lists can be very
large. Even in dimension two, linearity imposes a restriction beyond the counting bound above.

\subsection{Our Results}

Our main theorem bounds the output list size by $(1+\varepsilon)\ell/(1-\alpha)$
simultaneously for all input list sizes in a nearly exponential range. Here and throughout,
$\log$ denotes the natural logarithm.

\begin{theorem}[Main Theorem]\label{thm:main}
Fix $\alpha,\varepsilon\in(0,1)$ and a positive integer $d$. There exist a constant $\delta>0$,
an integer $n_0=O_{\alpha,\varepsilon}(d\log(2d))$, and a computable function
$f\colon\mathbb{Z}_{>0}\to\mathbb{Z}_{>0}$ such that the following holds. For every $n\ge n_0$
and every prime $p\ge f(n)$, a uniformly random $d$-dimensional subspace $C\subseteq\Fp^n$ is,
with probability at least $1-\varepsilon$,
\[
  \left(\alpha,\ell,\frac{1+\varepsilon}{1-\alpha}\ell\right)\text{-list recoverable}
  \qquad\text{for every integer }1\le\ell\le 2^{\delta n/\log n}.
\]
\end{theorem}

Our second theorem bounds the range of input list sizes for every linear code of dimension at
least two. The obstruction already occurs in a two-dimensional subcode, and the codewords
witnessing it agree with the lists at every coordinate.

\Needspace{8\baselineskip}
\begin{theorem}[Lower Bound for Large Input Lists]\label{thm:lower}
Fix $\alpha,\varepsilon\in(0,1)$. There exist a constant $c>0$, an integer $n_0$, and a
computable function $f\colon\mathbb{Z}_{>0}\to\mathbb{Z}_{>0}$ such that the following holds.
Let $n\ge n_0$, let $p\ge f(n)$ be prime, and let $C\subseteq\Fp^n$ be any linear code of
dimension at least two. For every integer $2^{cn}\le\ell\le p$, there are sets
$S_1,\ldots,S_n\subseteq\Fp$ of size $\ell$ for which
\[
  |C\cap(S_1\times\cdots\times S_n)|>\frac{1+\varepsilon}{1-\alpha}\ell .
\]
In particular, $C$ fails to be $(\alpha,\ell,(1+\varepsilon)\ell/(1-\alpha))$-list recoverable.
\end{theorem}

For fixed $d\ge 2$, the two results leave a factor of order $\log n$ between the exponents in
the input list size. The upper bound permits the field size to grow with $n$, and the same code
works for every stated input list size and every choice of the lists.

\subsection{Relation to Prior Work}

For ordinary list decoding ($\ell=1$), the list size of a random linear code has been studied
both near capacity and at very high error rates~\cite{GuruswamiHastadKopparty,Wootters}. With
several symbols per coordinate, the choice of field already matters in the absence of errors.
Guruswami et al.~\cite{GLMRSW} used a proper subfield to exhibit large output lists for random
linear codes near capacity. For list recovery from errors, Doron, Mosheiff, Resch, and
Ribeiro~\cite{DoronMosheiffReschRibeiro} showed that a random linear code over a fixed field has
output list size $O(1/\eta)$ with high probability when its rate is $\eta$ below list-recovery capacity
over that field. The implicit constant can depend on the field, radius, and input list size.
Here we fix the dimension and seek a bound linear in $\ell$ even as $\ell$ and the field size
grow with $n$.

Over large fields, Rudra and Wootters~\cite{RudraWootters} obtained bounds for random linear
codes under the stronger requirement that one bound the agreement fraction averaged over all the
candidate codewords. Li and Shagrithaya~\cite{LiShagrithaya} gave a bound on the output list
size independent of the field size and block length near capacity. They also showed that every
linear code of rate $R$ requires output list size $\ell^{\Omega(R/\tau)}$ in their parameter
range, where $\tau$ is the difference between $1-R$ and the error fraction. Brakensiek, Chen,
Dhar, and Zhang~\cite{BrakensiekChenDharZhang} bounded the entropy of a codeword in terms of the
entropies of its coordinates using discrete Brascamp--Lieb inequalities. This gave output list
sizes polynomial in $\ell$ for fixed $R$ and $\tau$. In our setting $R=d/n$ tends to zero and
$\tau=1-\alpha-R$ tends to a positive constant. We obtain a coefficient of $\ell$ arbitrarily
close to $1/(1-\alpha)$.

For codes of unrestricted dimension, related work optimizes the rate or error fraction with
the list sizes fixed. Resch, Yuan, and Zhang~\cite{ReschYuanZhang} determined the
error fraction at which positive rate becomes impossible for fixed alphabet size and fixed list
sizes. Brakensiek, Chen, Putterman, and Zhang~\cite{BCPZ} determined the optimal asymptotic rate
for unrestricted codes over sufficiently large alphabets when $\alpha,\ell$, and $L$ are
constant. Their rate expression $((L+1)(1-\alpha)-\ell)/L$ vanishes when $(L+1)(1-\alpha)=\ell$,
giving another way to see the scale of the counting bound. Theorem~\ref{thm:main} gives one
random linear code that approaches this scale over a growing range of $\ell$. For a survey of list
recovery, see Resch and Venkitesh~\cite{ReschVenkitesh}.

\section*{Acknowledgements and AI Use Methodology}
\begin{flushleft}
\renewcommand{\arraystretch}{1.3}
\begin{tabular}{@{}l@{\qquad}c@{}}
\textbf{Activity} & \textbf{AI used} \\
Asking the research question & [No] \\
Coming up with the approach & [Yes] \\
Proof development & [Yes] \\
Writing and exposition & [Yes] \\
Checking for bugs & [Yes] \\
Other supporting tasks & [Yes] \\
\end{tabular}
\end{flushleft}

The human authors obtained a less optimal form of the results of this paper entirely using human reasoning. We then posed our question -- without providing any information about our approach -- to the UCLA Moonshot AI for Math harness~\cite{MoonshotHarness} which made extensive use of queries to ChatGPT 5.5 Sol. 
The proof obtained by the harness was superior to our human-devised one both quantitatively and \emph{qualitatively}. (We also note that ChatGPT 5.5 Sol on its own was not able to prove this theorem when the authors attempted to query it directly.) After the human authors manually verified the all of the proofs, almost all of this paper was directly written by the harness, with some post-processing by ChatGPT 5.5 Sol and ChatGPT 6 Astra.

Of course, the human authors take full responsibility for this paper and its contents. If any reader is aware of references that should be added, please contact us and we will update the manuscript.

This research was supported in part by a Simons Investigator Award, DARPA expMath award, Laude
Moonshot grant, NSF grant 2333935, BSF grant 2022370, a Xerox Faculty Research Award, a Google
Faculty Research Award, an Okawa Foundation Research Grant, and the Symantec Chair of Computer
Science. This material is based upon work supported by the Defense Advanced Research Projects
Agency through Award HR001126CE054.

\section{Technical Overview}\label{sec:overview}

\MPS{}~\cite{MPS} studied reconstruction of a set of $b$ points in the real plane from the
multisets of their projections onto several lines. They showed that
$O(\log b\,\log\log(4b))$ suitably chosen directions suffice to determine the point set. Much
of our proof comes from their colored-graph and determinant arguments. The colored matchings
record equality of projections in their setting and equality of coordinate values in ours.
Spanning trees with disjoint color sets provide the nonzero determinants needed over a finite
field.

\paragraph{From Coordinate Lists to Matchings.}
In the reconstruction problem, suppose two distinct point sets give the same multiset of
projections in each direction. After canceling the common points, we can pair each point on one
side with a point on the other side that has the same projection. Each direction thus gives a
perfect matching, and we use the direction as its color. In list recovery, there is only one set
of messages, and some messages may miss a given list altogether. Nevertheless, a sufficiently
large output list gives us many pairs of messages that agree at a coordinate.

To see this, write an encoder as $a\mapsto(\lambda_1(a),\ldots,\lambda_n(a))$, where $a\in\Fp^d$
is the message and $\lambda_i\colon\Fp^d\to\Fp$ is the linear form at coordinate $i$. For any
positive integer $m$, write $[m]=\{1,\ldots,m\}$. Lund and Potukuchi~\cite[proof of
Theorem~10]{LundPotukuchi} also pair codewords that agree within a coordinate in their analysis
of randomly punctured codes. Here it is enough to count these pairs across all coordinates.

\begin{lemma}\label{lem:agreements}
Let $S_1,\ldots,S_n\subseteq\Fp$ have size $\ell$, and let $A\subseteq\Fp^d$ be a set of $b$
messages, each agreeing with the lists on at least $(1-\alpha)n$ coordinates. There is a
loopless multigraph on $A$ in which the edges of each color $i\in[n]$ form a matching, every
edge $\{u,v\}$ of color $i$ satisfies $\lambda_i(u)=\lambda_i(v)$, and the total number of edges
is at least
\[
                         \frac{n}{2}\bigl((1-\alpha)b-\ell\bigr).
\]
\end{lemma}

\begin{proof}
For each $i$, let $a_i$ count the messages $a\in A$ with $\lambda_i(a)\in S_i$. Group these
messages by the value of $\lambda_i(a)$, and form disjoint pairs within each group. There are at
most $\ell$ groups, so at most $\ell$ messages are left unpaired. This gives at least
$(a_i-\ell)/2$ edges of color $i$. Summing over $i$ and using $\sum_i a_i\ge (1-\alpha)nb$
proves the bound.
\end{proof}

\paragraph{From Matchings to Tree Equations.}
Suppose $b>(1+\varepsilon)\ell/(1-\alpha)$. The graph then has at least $\beta nb$ edges, where
$\beta=(1-\alpha)\varepsilon/(2(1+\varepsilon))$. Each color contributes at most one crossing
edge per vertex on the smaller side of a cut. Thus a lower bound on the number of crossing edges
also gives a lower bound on the number of crossing colors. The graph as a whole might be
disconnected. In Section~\ref{sec:graphs}, we find a subset of at least two vertices on which
the colors can be divided among $d$ groups so that each group contains a spanning tree. This
works as soon as $n$ is a sufficiently large multiple of $\log b\,\log\log(4b)$. This is the
source of the exponent $n/\log n$ in Theorem~\ref{thm:main}.

The disjointness of the color sets lets us choose a different coordinate form for each tree.
To illustrate this, take $d=2$. Fix the two colored trees $T_1$ and $T_2$, and
now treat the messages at their vertices as unknowns. For each color on $T_1$, choose the form
that reads the first coordinate of a message. For each color on $T_2$, choose the form that
reads the second. We can make both choices because no color occurs on both trees. The equations
on $T_1$ then force the first coordinate to be constant along the tree. The equations on $T_2$
do the same for the second coordinate. For this particular choice of forms, the only solutions
assign the same message to every vertex.

Writing the tree equations in the differences from one fixed message gives a square matrix.
The coordinate specialization above makes this matrix invertible, so its determinant is a
nonzero polynomial in the coefficients of the forms. The argument extends to dimension $d$ by
using one tree for each message coordinate. In Section~\ref{sec:equations}, we use the
Schwartz--Zippel lemma~\cite{Schwartz,Zippel} to show that, over a sufficiently large field,
a random choice of forms makes all the required determinants nonzero simultaneously. The union
bound runs over colored trees, even though the lists and messages producing those trees may be
chosen after the code is known.

Tree packings also arise in work of Guo, Li, Shangguan, Tamo, and Wootters on list recovery of
Reed--Solomon codes~\cite{GuoTree}. There the coordinate rows have the additional constraint
that they are evaluations of monomials at a single field element. Here the coefficients of a
coordinate form can be chosen independently, which is exactly what lets us assign all colors
from one tree to a single message coordinate. In subsequent work, Lester and S\o
rensen~\cite{LesterSorensen} used the graph and determinant lemmas of this paper to extend the
planar reconstruction bound to real point multisets in higher dimensions.

\section{Finding Spanning Trees with Disjoint Colors}\label{sec:graphs}

The graph we obtain by pairing messages may have isolated vertices. We therefore
look for a subset of at least two vertices on which several disjoint sets of
colors each give a connected graph. \MPS{}~\cite{MPS} used this idea to find two spanning trees
when every
color forms a perfect matching. Here we need $d$ trees, and counting agreements
with the input lists gives us a lower bound on the \emph{total} number of edges.
We allow parallel edges throughout, since two messages can agree in several
coordinates.

\begin{lemma}[Spanning Trees with Disjoint Colors]\label{lem:many-matchings}
Fix $\beta>0$ and an integer $d\ge 1$. There is a constant
$c=O_{\beta}(d\log(2d))$ with the following property. Let $G$ be a loopless
multigraph on $b\ge 2$ vertices with edges colored from $[m]$.
Suppose that the edges of each color form a matching, that $|E(G)|\ge\beta mb$,
and that
\[
                 m\ge c\log b\,\log\log(4b).
\]
Then some set $W$ of at least two vertices supports $d$ spanning trees in
$G[W]$ whose color sets are pairwise disjoint. A color may occur several times
in the same tree.
\end{lemma}

We follow the two steps of~\cite{MPS}. First we find an induced subgraph in
which many colors cross every cut, with a stronger bound when one side of the
cut is small. Then we divide the colors into random disjoint sets. To see what
cut bound suffices, we start with the second step. For a graph $H$ and a set
$I$ of colors, write $H_I$ for the graph on all the vertices of $H$ that keeps
just the edges with colors in $I$.

\begin{lemma}[Connectivity from Random Colors]\label{lem:random-colors}
Let $H$ be a loopless multigraph on $w\ge 2$ vertices with edges colored from
$[m]$. Suppose that for some $\gamma>0$ and every nonempty set $A$ of at most
$w/2$ vertices, at least $\gamma m\log(w/|A|)$ distinct colors occur on edges
from $A$ to its complement. Fix $0<\eta<1$. If $I$ is a uniformly random set
of $t$ colors, where
\[
       \frac{4}{\gamma}\bigl(\log\log(4w)+\log(1/\eta)+1\bigr)\le t\le m,
\]
then $H_I$ is connected with probability at least $1-\eta$.
\end{lemma}

\begin{proof}
Expose the colors in random order. Let $I_j$ be the first $j$ colors and
$X_j$ the number of components of $H_{I_j}$, counting isolated vertices.
When there are many components, most of them are small, and many colors leave
each small component. We use this to show that the expected drop in
$\log X_j$ is at least a $\gamma/4$ fraction of its current value.

Fix $I_j$ with $j<t$ and write $x=X_j$. Suppose $x\ge 2$. Call a component
small if it contains at most $w/2$ vertices. For a small component $A$, all
colors crossing from $A$ to its complement are still available, so the next
color crosses this cut with probability at least
$\gamma m\log(w/|A|)/(m-j)\ge\gamma\log(w/|A|)$. We claim that
\[
            \sum_{A\text{ small}}\log\frac{w}{|A|}
                       \ \ge\ \frac{x}{2}\log x .
\]
If all components are small, convexity of $-\log$ gives the stronger lower
bound $x\log x$, since the average component size is $w/x$. Otherwise there
are $x-1$ small components and their total size is at most $w/2$. The same
convexity bound gives $(x-1)\log(2(x-1))\ge (x/2)\log x$.

Consequently the next color crosses the cuts of at least
$(\gamma x/2)\log x$ small components in expectation. If it crosses the cuts
of $h$ small components, then the component count drops by at least $h/2$.
Indeed, contract each old component and look at the edges of the new color
between the resulting vertices. A connected group of $k\ge 2$ vertices
accounts for a drop of $k-1\ge k/2$, and every one of those $h$ components
belongs to such a group. Thus
$\mathbb E[x-X_{j+1}\mid I_j]\ge(\gamma x/4)\log x$. Since
$\log(x/X_{j+1})\ge (x-X_{j+1})/x$, this gives
\[
       \mathbb E[\log X_{j+1}\mid I_j]\le (1-\gamma/4)\log X_j.
\]
The inequality also holds when $X_j=1$, since the graph then stays connected.

A singleton cut gives $\gamma\log w\le 1$, so $1-\gamma/4$ is positive.
Starting with $X_0=w$ and iterating, we get
$\mathbb E[\log X_t]\le e^{-\gamma t/4}\log w$. A disconnected graph has
$\log X_t\ge\log 2$, so Markov's inequality gives
\[
     \Pr(H_I\text{ is disconnected})
       \le \frac{e^{-\gamma t/4}\log w}{\log 2}
       \le \eta\,\frac{\log w}{e\log 2\,\log(4w)}
       \le \eta. \qedhere
\]
\end{proof}

\begin{proof}[Proof of Lemma~\ref{lem:many-matchings}]
It remains to find a vertex set on which we can apply
Lemma~\ref{lem:random-colors}. For $U\subseteq V(G)$, let $e(U)$ count the
edges with both endpoints in $U$, including multiplicity. Choose $W$ of size
$w\ge 2$ maximizing the ratio
\[
       \frac{e(U)}{|U|\log |U|}\qquad (|U|\ge 2),
\]
and let $\rho$ be its value at $W$. This choice of density is also used
by Wang~\cite{Wang} to obtain expansion in edge-colored graphs. Taking
$U=V(G)$ gives $\rho\ge\beta m/\log b$.

Consider a cut of $W$ with one side $A$ of size $1\le a\le w/2$. The two sides
span at most $\rho a\log a$ and $\rho(w-a)\log(w-a)$ edges, respectively.
The bound also holds for a singleton since $G$ has no loops. Subtracting from
$e(W)=\rho w\log w$, we find that the cut has at least
\[
       \rho\left(a\log\frac{w}{a}+(w-a)\log\frac{w}{w-a}\right)
          \ \ge\ \rho a\log\frac{w}{a}
\]
edges. Each color contributes at most $a$ edges because it is a matching.
Thus at least $(\beta m/\log b)\log(w/a)$ distinct colors cross the cut.
This gives the hypothesis of Lemma~\ref{lem:random-colors} on $G[W]$ with
$\gamma=\beta/\log b$.

We want each of $d$ disjoint sets of colors to connect $W$ with failure
probability at most $1/(2d)$. Since $w\le b$, the lemma allows us to use
blocks of size
\[
          t=\left\lceil\frac{4\log b}{\beta}
                   \bigl(\log\log(4b)+\log(2d)+1\bigr)\right\rceil.
\]
For $b\ge 2$ both $\log b$ and $\log\log(4b)$ are bounded below by positive
absolute constants. It follows that $dt\le c\log b\,\log\log(4b)$ for a
choice of $c=O_{\beta}(d\log(2d))$. Thus $dt\le m$. Take the first $d$ blocks
of $t$ colors in a uniformly random permutation of $[m]$. Each block is a
uniform $t$-element subset, so a union bound gives probability at least $1/2$
that all $d$ blocks connect $W$. Choosing a spanning tree from each block
gives the desired trees.
\end{proof}

\section{The Linear Equations on the Trees}\label{sec:equations}

The edge equations only involve differences between messages. We may therefore place one vertex
at the origin. Proving that every other vertex is also at the origin is then equivalent to
proving that all the messages are equal. On $w$ vertices this leaves $d(w-1)$ scalar unknowns.
Each of our $d$ trees supplies $w-1$
equations, so the system is square. \MPS{}~\cite[Lemma~2.3]{MPS} use a determinant to analyze
the system for two trees in the plane. We give the argument for $d$ trees over any field.

\begin{lemma}[The Determinant of the Tree Equations]\label{lem:determinant}
Let $F$ be a field and let $T_1,\ldots,T_d$ be spanning trees on $[w]$, where
$w\ge 2$ and $d\ge 1$. Their edges have colors in $[m]$, and different trees
use disjoint sets of colors. For each $i\in[m]$, let
$\lambda_i(z)=\sum_{j=1}^d a_{i,j}z_j$ have indeterminate coefficients.
Orient the edges of the trees arbitrarily. For points $q_1,\ldots,q_w\in F^d$,
consider the equations $\lambda_i(q_u-q_v)=0$ for the oriented edges $(u,v)$
of color $i$. Write these equations in the $d(w-1)$ coordinates of the
differences $q_s-q_w$, $s\in[w-1]$. Their coefficient matrix $R$ has
one row per edge and one column for each $(s,j)\in[w-1]\times[d]$, with
\[
     R_{e,(s,j)}=(\mathbf 1_{s=u}-\mathbf 1_{s=v})\,a_{i,j}
     \qquad\text{for an edge $e=(u,v)$ of color $i$.}
\]
Then $\det R$ is a nonzero polynomial over $F$ of degree at most $d(w-1)$.
For any choice of the coefficients in $F$ where $\det R\ne 0$, the edge
equations imply $q_1=\cdots=q_w$.
\end{lemma}

\begin{proof}
There are $d(w-1)$ edges in the trees, so $R$ is square. To see that its
determinant is nonzero, set $\lambda_i(z)=z_j$ whenever color $i$ occurs
in $T_j$, and set unused forms to zero. In this system every edge of $T_j$
equates the $j$th coordinates of its endpoints. Since $T_j$ spans $[w]$,
$(q_s-q_w)_j=0$ for all $s$. This holds for each $j$, so the specialized
matrix has trivial kernel (we may take $q_w=0$ to realize any vector of
differences). Its determinant is nonzero over $F$, as claimed.

Each row of $R$ is linear in the coefficients, so the determinant has degree
at most $d(w-1)$. For any coefficients at which it is nonzero, the vector of
differences satisfying the edge equations is in the kernel of an invertible
matrix. Hence every point equals $q_w$.
\end{proof}

\paragraph{One Choice of Forms for All the Trees.}
The sets of messages and input lists may be chosen after the code is known. To use
Lemma~\ref{lem:determinant}, we need its conclusion simultaneously for every collection of trees
that could arise. Once the number of vertices is bounded, we can enumerate the colored trees
without referring to any field elements.

\begin{proposition}\label{prop:random-good}
Choose $n$ linear forms on $\Fp^d$ by choosing their coefficients independently and uniformly
from $\Fp$. Fix an integer $B\ge 2$. Except with probability at most
\[
                      \frac{dB^2(B^2 n)^{dB}}{p},
\]
the following holds for every $2\le w\le B$ and every $d$-tuple of spanning trees on $[w]$ with
colors in $[n]$ and pairwise disjoint color sets. The only assignments of messages to the
vertices satisfying all the tree equations of Lemma~\ref{lem:determinant} are the constant
assignments.
\end{proposition}

\begin{proof}
Fix the number $w$ of vertices. A colored tree can be described by listing two endpoints and a
color for each of its $w-1$ edges. Thus there are at most $(B^2 n)^{d(w-1)}$ choices for the $d$
trees, and at most $B(B^2 n)^{dB}$ choices over all $w$. For any one choice with disjoint color
sets, Lemma~\ref{lem:determinant} gives a nonzero determinant polynomial of degree at most
$d(w-1)$. The Schwartz--Zippel lemma~\cite{Schwartz,Zippel} bounds the probability that a
nonzero polynomial of total degree $D$ vanishes on independent uniform field elements by $D/p$.
Applying it to each determinant and summing the failure probabilities proves the proposition.
\end{proof}

\section{Proof of the Main Theorem}\label{sec:upper}

\begin{proof}[Proof of Theorem~\ref{thm:main}]
Write $K=(1+\varepsilon)/(1-\alpha)$, and take $\beta=(1-\alpha)\varepsilon/(2(1+\varepsilon))$
as in Section~\ref{sec:overview}. If list recovery fails at input size $\ell$, we can keep
exactly $b=\lfloor K\ell\rfloor+1$ distinct codewords that agree with the lists.
Lemma~\ref{lem:agreements} gives at least $\beta nb$ edges among their messages. To apply
Lemma~\ref{lem:many-matchings} with $n$ colors, we therefore need
\begin{equation}\label{eq:colors-needed}
                       n\ge c\log b\,\log\log(4b),
\end{equation}
where we take $c\ge 1$ to be the constant from that lemma. In particular
$c=O_{\alpha,\varepsilon}(d\log(2d))$.

\smallskip\noindent\textbf{Bounding the Input List Size.}
For $n\ge3$, take $0<\delta\le 1$ and put $x=\delta n/\log n$. If $\ell\le 2^x$, then $b\le
K2^x+1$. For $x<1$ this bounds $\log b\,\log\log(4b)$ by a constant depending only on $K$. For
$x\ge1$ it gives $\log b=O_K(x)$ and $\log\log(4b)=O_K(\log n)$, since $x\le n$. In both cases
there is a constant $a\ge1$ depending only on $\alpha,\varepsilon$ such that
\[
                       \log b\,\log\log(4b)\le a(1+\delta n).
\]
To make the field-size threshold computable using integer arithmetic, choose an integer $k\ge
2ac$ and take $\delta=(\log 2)/k$, so that $ac\delta\le 1/2$. Set $n_0=\max\{3,d+1,\lceil
2ac\rceil\}$. Then~\eqref{eq:colors-needed} holds for all $n\ge n_0$ and all $\ell$ in the
claimed range. This choice also gives $n_0=O_{\alpha,\varepsilon}(d\log(2d))$.

\smallskip\noindent\textbf{Choosing the Random Code.}
We now choose the field large enough that one draw of the coordinate forms works for all the
lists. Fix an integer $h>\max\{K+1,2/\varepsilon\}$ and, for $n\ge3$, set
\[
              B=h2^{\lceil\delta n/\log n\rceil},\qquad f(n)=h d B^2(B^2n)^{dB}.
\]
Here $B$ bounds $K2^x+1$ and hence every witness size $b$ in our range. The function $f$ is
computable. Indeed, the exponent in $B/h$ is the least integer $j\ge0$ with $n^{kj}\ge 2^n$, and
we may set $f(1)=f(2)=1$. Choose a prime $p\ge f(n)$ and a uniform $n\times d$ matrix $M$ over
$\Fp$. Its rows give the $n$ independent forms. By Proposition~\ref{prop:random-good}, the
probability that any of the tree systems on at most $B$ vertices admits a nonconstant assignment
is at most $1/h\le\varepsilon/2$.

Condition on $M$ having full column rank. For each fixed nonzero $v\in\Fp^d$, $Mv$ is uniform in
$\Fp^n$, so the probability that $M$ has a nonzero kernel vector is at most $(p^d-1)p^{-n}\le
1/p$, since $n\ge d+1$. Thus the conditioning event has probability at least $1/2$. Every
$d$-dimensional subspace of $\Fp^n$ has the same number of ordered bases. Thus, under this
conditioning, $C=\operatorname{im}M$ is a uniformly random $d$-dimensional code, and the
conditional probability of any failure in Proposition~\ref{prop:random-good} is at most
$\varepsilon$.

Suppose the conclusion of Proposition~\ref{prop:random-good} holds. If some lists had more than
$K\ell$ agreeing codewords, the $b$ codewords chosen above would have distinct preimages under
$M$. By~\eqref{eq:colors-needed} and Lemma~\ref{lem:many-matchings}, their graph contains $d$
spanning trees on a common set of at least two of these messages, with disjoint color sets. Each
edge equates the values of its coordinate form at its endpoints. The proposition therefore
forces the messages at those vertices to be equal, a contradiction. This argument applies to
every $\ell\le 2^{\delta n/\log n}$ on the same event.
\end{proof}

\section{The Lower Bound for Exponential Input Lists}\label{sec:lower}

The cube construction in discrete tomography takes one vector in each of several prescribed
directions and forms all their subset sums~\cite{BianchiLonginetti,MPS}. A linear form that
vanishes on one of these vectors gives the same value whether we include that vector or leave it
out. We use the same observation in a two-dimensional subcode. We allow more than two choices in
each direction so that the construction works for every input list size in the range of
Theorem~\ref{thm:lower}. Lund and Potukuchi~\cite[Theorem~11]{LundPotukuchi} used sums of
arithmetic progressions to produce large zero-error output lists in randomly punctured
Reed--Solomon codes of dimension at least two. Li and
Shagrithaya~\cite[Section~5]{LiShagrithaya} used a box of linear combinations of codewords with
disjoint supports on selected coordinates. In our construction, the step in each direction can
be chosen after the code is fixed.

\begin{proof}[Proof of Theorem~\ref{thm:lower}]
Fix $\alpha,\varepsilon\in(0,1)$ and write $K=(1+\varepsilon)/(1-\alpha)$. Choose an integer
$T>4(K+1)$, and take $n_0=2$, $c=\log_2 T$, and $f(n)=T^n$. Let $n\ge 2$, let $p\ge f(n)$ be
prime, and fix any integer $T^n\le\ell\le p$. We will find more than $K\ell$ codewords that
agree with the same lists in \emph{every} coordinate.

Take a two-dimensional subcode of $C$ and identify its messages with $\Fp^2$. Write
$\lambda_i\colon\Fp^2\to\Fp$ for its $i$th coordinate map. Since distinct messages give distinct
codewords, the common kernel of these maps is $\{0\}$. The nonzero coordinate maps therefore
have at least two distinct kernels. Choose one nonzero vector $u_j$ on each distinct kernel
line, giving $r$ vectors $u_1,\ldots,u_r$ with $2\le r\le n$.

\smallskip\noindent\textbf{Choosing the Side Lengths.}
We will take all messages of the form $\sum_{j=1}^r a_j s_j u_j$, where each $a_j$ ranges over
the integers $0,\ldots,t_j-1$ and $s_j\in\Fp$ is a nonzero scalar that we may choose. There are
$P=\prod_{j=1}^r t_j$ choices of coefficients. A coordinate map with $u_j$ in its kernel takes
at most $P/t_j$ values on these messages, regardless of the scalars. We thus want the product to
be just above $K\ell$, with every side long enough that removing it leaves at most $\ell$
choices. Fix all but the first side to have length $T$ and take the smallest first side that
gives the desired product,
\[
  t_2=\cdots=t_r=T,\qquad
  t_1=\left\lfloor\frac{K\ell}{T^{r-1}}\right\rfloor+1.
\]
Since $\ell\ge T^n\ge T^r$ and $K>1$, every side has length at least $T$. Moreover,
\[
 K\ell<P\le K\ell+T^{r-1}\le (K+1)\ell.
\]
It follows that $P/t_j\le (K+1)\ell/T<\ell$ for every $j$.

\smallskip\noindent\textbf{Making the Messages Distinct.}
Choose $s_1,\ldots,s_r$ independently and uniformly from $\Fp\setminus\{0\}$. If two coefficient
choices give the same message, their difference is a nonzero integer vector $b=(b_1,\ldots,b_r)$
with $|b_j|<t_j$ and
\[
 \sum_{j=1}^r b_j s_j u_j=0.
\]
Each $t_j$ is smaller than $p$. Indeed, $T<p$ since $p\ge T^2$, and
$t_1=P/T^{r-1}\le(K+1)p/T<p$. Thus every nonzero $b_j$ stays nonzero in $\Fp$. A difference with
just one nonzero entry cannot give a collision. For a fixed difference with at least two nonzero
entries, choose two such indices. The corresponding $u_j$ lie on distinct lines, so once we fix
all the other scalars there is at most one pair of values for these two scalars that solves the
equation. The probability of this collision is at most $1/(p-1)^2$.

There are fewer than $\prod_{j=1}^r 2t_j=2^rP$ possible differences. By a union bound, the
probability of any collision is at most
\[
 \frac{2^rP}{(p-1)^2}
 \le \frac{2^{n+2}(K+1)}{p}<1,
\]
where the last inequality uses $p\ge T^n$ and $(T/2)^n\ge (T/2)^2>4(K+1)$. Fix scalars for which
all $P$ messages are distinct and let $A$ be the set of their codewords. Each nonzero coordinate
map annihilates one of the $u_j$, so the codewords in $A$ take at most $P/t_j\le\ell$ values
there. A zero coordinate map takes only the value zero. At each coordinate we can therefore
include all these values in a list of size $\ell$, since $\ell\le p$. Every codeword in $A$
agrees with every list, and $|A|=P>K\ell$, as desired.
\end{proof}

\bibliographystyle{alpha}
\bibliography{refs}
\end{document}